\documentclass[12pt]{article} 
\usepackage{amsmath}
\usepackage{amssymb}
\usepackage{url}

\usepackage{graphicx}
\usepackage{bm}

 \usepackage{mathptmx}      

\usepackage{latexsym}
\usepackage{amsmath}
\usepackage{amssymb}
\usepackage{url}
\usepackage{comment}

\newcommand{\thSAT}{{\rm 3SAT}}
\newcommand{\st}{\mathrel{:}}

\newcommand{\DTIME}{{\rm DTIME}}
\newcommand{\NTIME}{{\rm NTIME}}

\newcommand{\YES}{{\rm YES}}

\newcommand{\NO}{{\rm NO}}
\newcommand{\GCE}{{\rm GCE}}
\newcommand{\GC}{{\rm GC}}

\newcommand{\xbar}{\overline{x}}
\newcommand{\xonebar}{\overline{x}_1}
\newcommand{\xtwobar}{\overline{x}_2}
\newcommand{\xthreebar}{\overline{x}_3}
\newcommand{\xfourbar}{\overline{x}_4}
\newcommand{\xibar}{\overline{x}_i}
\newcommand{\xnbar}{\overline{x}_n}

\newcommand{\OBS}{{\rm OBS}}
\newcommand{\COL}{{\rm COL}}

\newcommand{\penp}{\hbox{P$=$NP}}
\newcommand{\SAT}{{\rm SAT}}
\newcommand{\NP}{{\rm NP}}
\renewcommand{\P}{{\rm P}}

\newcommand{\N}{{\mathbb{N}}}

\newcommand{\es}{\emptyset}

\newcommand{\GNM}{G_{N,M}}
\newcommand{\GNN}{G_{N,N}}

\newcounter{savenumi}

\newtheorem{theoremfoo}{Theorem}[section] 
\newenvironment{theorem}{\pagebreak[1]\begin{theoremfoo}}{\end{theoremfoo}}

\newtheorem{lemmafoo}[theoremfoo]{Lemma}
\newenvironment{lemma}{\pagebreak[1]\begin{lemmafoo}}{\end{lemmafoo}}
\newtheorem{conjecturefoo}[theoremfoo]{Conjecture}

\newtheorem{conventionfoo}[theoremfoo]{Convention}

\newtheorem{porismfoo}[theoremfoo]{Porism}

\newtheorem{gamefoo}[theoremfoo]{Game}

\newtheorem{corollaryfoo}[theoremfoo]{Corollary}

\newtheorem{openfoo}[theoremfoo]{Open Problem}

\newtheorem{exercisefoo}{Exercise}

\newcommand{\fig}[1] 
{
\begin{figure}
\begin{center}
\input{#1}
\end{center}
\end{figure}
}

\newtheorem{potanafoo}[theoremfoo]{Potential Analogue}

\newtheorem{potthmfoo}[theoremfoo]{Potential Theorem}

\newtheorem{notefoo}[theoremfoo]{Note}

\newtheorem{notabenefoo}[theoremfoo]{Nota Bene}

\newtheorem{nttn}[theoremfoo]{Notation}

\newtheorem{empttn}[theoremfoo]{Empirical Note}

\newtheorem{examfoo}[theoremfoo]{Example}

\newtheorem{dfntn}[theoremfoo]{Def}
\newenvironment{definition}{\pagebreak[1]\begin{dfntn}\rm}{\end{dfntn}}

\newtheorem{propositionfoo}[theoremfoo]{Proposition}

\newenvironment{proof}
   {\pagebreak[1]{\narrower\noindent {\bf Proof:\quad\nopagebreak}}}{\QED}

\newcommand{\yyskip}{\penalty-50\vskip 5pt plus 3pt minus 2pt}
\newcommand{\blackslug}{\hbox{\hskip 1pt
       \vrule width 4pt height 8pt depth 1.5pt\hskip 1pt}}
\newcommand{\QED}{{\penalty10000\parindent 0pt\penalty10000
       \hskip 8 pt\nolinebreak\blackslug\hfill\lower 8.5pt\null}
       \par\yyskip\pagebreak[1]}

\newcommand{\BBB}{{\penalty10000\parindent 0pt\penalty10000
       \hskip 8 pt\nolinebreak\hbox{\ }\hfill\lower 8.5pt\null}
       \par\yyskip\pagebreak[1]}

\newtheorem{factfoo}[theoremfoo]{Fact}

\begin{document}

\centerline{\bf The Complexity of Grid Coloring}

\centerline{\bf by} 

\centerline{Daniel Apon\footnote{University of Maryland at College Park, MD, 20742
              {\tt dapon.crypto@gmail.com}}}

\centerline{William Gasarch\footnote{University of Maryland at College Park, MD, 20742 {\tt gasarch@umd.edu}}}

\centerline{Kevin Lawler\footnote{\tt kevin@permenentco.com}}

\begin{abstract}
A $c$-coloring of the grid $\GNM=[N]\times [M]$ is a mapping of $\GNM$ into
$[c]$ such that no four corners forming a rectangle have the same color.
In 2009 a challenge was proposed to find a 4-coloring of $G_{17,17}$.
Though a coloring was produced, finding it proved to be difficult.
This raises the question of whether there is some complexity lower bound.
Consider the following problem:
given a partial $c$-coloring of the  $\GNM$ grid, can it be extended
to a full $c$-coloring? We show that this problem is NP-complete.
We also give a Fixed Parameter Tractable algorithm for this problem with parameter $c$.
\end{abstract}

\section{Introduction}

\begin{definition}
\begin{enumerate}
\item
If $x\in\N$ then $[x]$ denotes the set
$\{1,\ldots,x\}$.
$\GNM$ is the set $[N]\times[M]$.
\end{enumerate}
\end{definition}

\begin{definition}
Let $n,m,c\ge 1$. 
\begin{enumerate}
\item
A \emph{rectangle of $\GNM$} is a subset
of the form
$$\{(a,b),(a+d_1,b),(a+d_1,b+d_2),(a,b+d_2)\}\subseteq \GNM$$
for some $a, b, d_1, d_2 \in \N$
with $d_1,d_2\ge 1$ 
Note that we are only looking at the four corners of the rectangle---nothing else.
\item
Let $\chi: \GNM\rightarrow [c]$. A {\it monochromatic rectangle}
is a rectangle where all 4 elements of it are colored the same. 
\item
Let $\chi: \GNM\rightarrow [c]$. 
If there are no rectangles with all four corners the same color
then we call $\chi$ a {\it $c$-coloring}.
If the $c$ is understood we may just say \emph{a coloring}. 
We sometimes use the term {\it proper $c$-coloring} rather than {\it $c$-coloring}
to stress the fact that it has no monochromatic rectangles.
\item 
A grid $\GNM$ is \emph{$c$-colorable} if there is a $c$-coloring of it. 
\end{enumerate}
\end{definition}

Fenner et al.~\cite{grid} explored the following problem:

\centerline{\it Which grids are $c$-colorable for a given fixed $c$?}

\section{History and Our Results} 

We state some of the results of Fenner et al.~\cite{grid}.

\begin{enumerate}
\item
For all $c\ge 2$, $G_{c+1,c\binom{c+1}{2}+1}$ is not $c$-colorable.
\item
For all $c$ there exists a finite number of grids, denoted $\OBS_c$,  such that
$\GNM$ is $c$-colorable iff it doesn't contain any element of $\OBS_c$. 
$\OBS$ stands for {\it the obstruction set}.
\item
$\OBS_2 = \{ G_{3,7}, G_{5,5}, G_{7,3} \}$.
This was obtained without the aid of a computer program.
\item
$\OBS_3 = \{G_{19,4}, G_{16,5}, G_{13,7}, G_{11,10}, G_{10,11}, G_{7,13}, G_{5,16}, G_{4,19} \}$.
A computer aided search was used to find a 3-coloring of $G_{10,10}$.
\item
$$\OBS_4=
\{ G_{41,5},
   G_{31,6},
   G_{29,7},
   G_{25,9},
   G_{23,10},
   G_{22,11},
   G_{21,13},
   G_{19,17} \} \bigcup $$

$$\{
   G_{17,19},
   G_{13,21},
   G_{11,22},
   G_{10,23},
   G_{9,25},
   G_{7,29},
   G_{6,31},
   G_{5,41}\}$$
The authors were stuck for a long time trying to find 4-colorings of
$G_{17,17}$, $G_{17,18}$, $G_{18,18}$, $G_{12,21}$, and $G_{10,22}$
(we omit the symmetric cases which follow automatically, i.e., if there is a
4-coloring of $G_{22,10}$ then there is one for and $G_{10,22}$).
They believed these were all 4-colorable.
William Gasarch put a bounty of $17^2=289$ dollars for
a 4-coloring of $G_{17,17}$ and posted this challenge to ComplexityBlog~\cite{Gasarch-2009}.
Bernd Steinbach and
Christian Posthoff
found 4-colorings of $G_{17,17}$, $G_{18,18}$, and $G_{12,21}$ and
received the reward.
Brad Larsen found a 4-coloring of $G_{22,10}$.
Brad Larsen posted the 4-coloring
saying he used a $\SAT$-solver but he did not elaborate.
Steinbach and Posthoff published their results and their methods.
In brief, they used a very deep analysis that allowed for a strong reduction of the problem,
and then used the Universal SAT-Solver clasp.
See their articles~\cite{SP-2012a,SP-2012b,SP-2012c,SP-2013b,SP-2013a,SP-2015} and
a book edited by Steinbach~\cite{Steinbach-2014} that has several chapters
explaining how they found a 4-coloring of $G_{12,21}$ in detail.
These results completed the search for $\OBS_4$. 
\item
Finding $\OBS_5$ seems to be beyond current technology.
\end{enumerate}

The difficulty of 4-coloring $G_{17,17}$ and pinning down $\OBS_5$ raise
the following question: is the problem of grid coloring hard?
In Section~\ref{se:gce} we define the {\it Grid Coloring Extension Problem}.
In Section~\ref{se:npc} we show this problem is $\NP$-complete.
Does this really indicate that  4-coloring is hard?
In Section~\ref{se:no} we discuss the issue.
In Section~\ref{se:fpt} we show that the grid coloring extension problem is
Fixed Parameter Tractable with parameter $c$.
Does this really give a way to find extensions quickly?
In Section~\ref{se:real} we discuss this issue. 
In Section~\ref{se:open} we present open problems.

\section{Definition of the Grid Coloring Extension Problem}\label{se:gce}

\begin{definition}
Let $N,M,c\in\N$.
\begin{enumerate}
\item
Given a grid $\GNM$, a {\it cell} is an element $(i,j)\in \GNM$. 
\item
A \emph{partial $c$-coloring $\chi$ of $\GNM$} is
a mapping of a subset of $\GNM$ to $[c]$ with no monochromatic
rectangle on the points where it is defined. 
See Figure~\ref{fig:partialcol} for an example.
\item
If $\chi$ is a partial $c$-coloring of $\GNM$ then
$\chi'$ is \emph{an extension of $\chi$} if
$\chi'$ is a partial $c$-coloring of $\GNM$ which
\begin{enumerate}
\item
is defined on every cell that $\chi$ is defined,
\item 
agrees with $\chi$ on those cells,
\item 
may be defined on more cells, and 
\item
has no monochromatic rectangles. 
\end{enumerate}
\item
A \emph{total mapping $\chi$ of $\GNM$ to $[c]$} is
a mapping of $\GNM$ to $[c]$.  This would normally just
be called a mapping, but we use the term total to distinguish it
from a partial mapping.
\end{enumerate}
\end{definition}

\begin{figure}
\begin{large}
\[
\begin{array}{|c||c|c|c|c|c|c|c|c|c|c|}
\hline
R &    & &   &   &   & & &  \cr
\hline
R &   &&  &   &   & & &  \cr
\hline
B     &   & R &   &   &   & & &  \cr
\hline
R &   & &   & B &   & & &  \cr
\hline
R &   & &   &   & B & & &  \cr
\hline
R &   & &   &   &   & & &  \cr
\hline
\hline
R & R &R& R & R &R  &R&R&R \cr
\hline
\end{array}
\]
\caption{Example of a partial coloring of $G_{9,7}$}
\label{fig:partialcol}
\end{large}
\end{figure}

\begin{definition}
$\GCE$ is the following problem:
\begin{itemize}
\item
{\it Input}  $N,M,c\ge 1$ and $\chi$ a partial $c$-coloring of $\GNM$.
The numbers $N,M,c$ are in unary. 
\item
{\it Output} YES if there is an extension of $\chi$ to a total $c$-coloring of 
$\GNM$, NO otherwise.
\end{itemize}
$\GCE$ stands for {\it Grid Coloring Extension}.
\end{definition}

We show that $\GCE$ is $\NP$-complete.
This result may explain why the original $17\times 17$ challenge
was so difficult.
Then again---it may not.
We discuss this further in Section~\ref{se:no}.

\section{$\GCE$ is $\NP$-complete}\label{se:npc}

Before showing that $\GCE$ is $\NP$-complete we briefly discuss its
status within $\NP$.
We first state and prove an easy upper bound.

\begin{theorem}
$\GCE\in\NTIME(O(N^2M^2))$ with certificate of size $O(NM\log c)$. 
\end{theorem}

\begin{proof}
Here is a nondeterministic algorithm

\begin{enumerate}
\item 
{\it Input} $(N,M,c,\chi)$.
\item
Guess an extension $\chi'$ of the $c$-coloring $\chi$ to a total mapping of $\GNM$ to $[c]$.
Note that $\chi'$, the certificate, is of size $O(NM\log c)$. 
\item 
For all 
$$\{(a,b),(a+d_1,b),(a+d_1,b+d_2),(a,b+d_2)\}\subseteq \GNM$$
do the following
\begin{enumerate}
\item
Check if 

$$\chi'(a,b)=\chi'(a+d_1,b)=\chi'(a+d_1,b+d_2)=\chi'(a,b+d_2).$$

\item
If yes then this branch stops and {\it outputs} NO.
\item
If no then (a) if this is the last rectangle to check then stop and {\it output} YES,
(b) if not then proceed to the next rectangle. 
\end{enumerate}
\end{enumerate}

Each execution of the loop body takes $O(1)$ time.
To get a time bound we need an upper bound on how often the loop is
executed. This is upper bounded by the number of  rectangles.

The number of ways to pick $a$ is $N$.
The number of ways to pick $b$ is $M$.
The number of ways to pick $d_1$ is $\le N$.
The number of ways to pick $d_2$ is $\le M$.
Hence the number of rectangles is $O(N^2M^2)$. 
Hence the runtime of any one branch is $O(N^2M^2)$.
Hence the algorithm is in $\NTIME(O(N^2M^2))$. 
\end{proof}

We obtain a better bound.
Kreveld and De Berg~\cite{KB-1991} proved, in our notation,  the following lemma.

\begin{lemma}\label{le:KB}
There is an algorithm that will, given a set of cells $P\subseteq \GNM$,
determine if $P$ contains a rectangle, in time $O((NM)^{3/2})$.
\end{lemma}

From this result we show the following:

\begin{theorem}\label{th:innp}
$\GCE\in\NTIME(O(c(MN)^{3/2}))$.
\end{theorem}

\begin{proof}
Given $(N,M,c,\chi)$ the witness is a proposed extension $\chi'$ of $\chi$ to a $c$-coloring
of $\GNM$. The following algorithm tries to verifies that $\chi'$ is a coloring.

\begin{enumerate}
\item
{\it Input} $(N,M,c,\chi,\chi')$. We assume the colors are $\{1,\ldots,c\}$.
\item
Verify that $\chi'$ is an extension of $\chi$. This takes $O(NM)$ steps.
\item
For all $1\le i\le c$
\begin{enumerate}
\item
Let $P$ be the set of cells colored $i$. It takes $O(MN)$ time to identify $P$.
\item
Use the algorithm from Lemma~\ref{le:KB} to determine if $P$ contains a rectangle.
This takes time $O((NM)^{3/2})$.
If the algorithm says $P$ contains a rectangle then
{\it output} NO and stop. Otherwise proceed to the next $i$.
\end{enumerate}
\item
(If the algorithm got here then, for all $1\le i\le c$, there is no $i$-colored
rectangle.)
{\it Output} YES and stop.
\end{enumerate}

The time spent in the For-Loop dominates everything else. That time is clearly $O(c(NM)^{3/2})$.
\end{proof}

We make one observation about $\GCE$ and $\SAT$ before our proof.
It is an easy exercise to express the question $(N,M,c,\chi)\in \GCE$ as a SAT formula.
(This was the starting point for the work of Steinback and Posthoff with $\chi$
being the empty function.)
This shows that $\GCE$ reduces to $\SAT$ but not that $\SAT$ reduces to $\GCE$.
Hence this reduction does not help us obtain a lower bound on the complexity of $\GCE$.  

We now show $\GCE$ is $\NP$-complete.

\begin{theorem}\label{th:npc}
$\GCE$ is $\NP$-complete.
\end{theorem}

\begin{proof}

By Theorem~\ref{th:innp}, $\GCE\in\NP$.

We give a reduction of $\thSAT$ to $\GCE$.
The input will be a 3CNF formula
$$\phi(x_1,\ldots,x_n)=C_1\wedge \cdots\wedge C_m$$
with $n$ free variables and $m$ clauses.
The output will be  $(N,M,c,\chi)$ where
\begin{itemize}
\item
$N,M,c\in\N$,
\item
$\chi$ is a partial $c$-coloring of $\GNM$, and
\item
$\phi\in \thSAT\hbox{ iff } (N,M,c,\chi)\in \GCE.$
\end{itemize}

We can assume that $\phi$ never has a clause that contains either
(1) the same literal twice, or
(2) a variable and its negation.
Condition (1) will be needed in Part III of the construction.
Condition (2) will be needed in the proof of Claim 4.

The reduction we show you {\it does not quite work!}; however,
it has most of the ideas needed.
There is a problem with it that will be revealed when we try to prove
Claim 4. During that proof we will see what goes wrong and modify the
construction so that Claim 4 is true.

Visualize the full grid as a core subgrid with additional entries to the left and below.
These additional entries are there to enforce that some colors in the core grid occur only once.

\newpage


{\bf Conventions}
\begin{enumerate}
\item
 Throughout this proof {\it extension} means
{\it an extension that uses the colors $T$,$F$ on some of the uncolored cells and
does not have a monochromatic rectangle}. 
It may or may not extend to the entire grid.
\item
In our figures we will have literals labeling some of the rows and clauses labeling some of
the columns. These are not part of the construction. The literals and clauses are visual aids. We may refer to
{\it row $x_7$} or {\it column $C_3$}.
\item
In our figures we will have double lines to separate things. These lines are not part of the construction.
These are visual aids.
\item
The colors will be $T$, $F$, and some of the $(i,j)\in \GNM$. Many of the cells that
are in the core grid will be colored $(i,j)$ where that is their position in the core
grid. In the figures we will denote the color by $D$ for distinct. Part I of the construction
will make sure that no other cell in the core grid can have that color.
\end{enumerate}

\bigskip

The reduction is in four parts.
We will mainly construct a core grid which will be $2n+m$ by $2n+2m+1$
(when we later modify the construction the core grid will be bigger, though still linear in $n,m$).

In all figures the left bottom cell of the core grid is indexed $(1,1)$.

\noindent
{\bf Part I: Forcing a color to appear only once in the core grid.}

For $(i,j)$ in the core grid  we will often set $\chi(i,j)$ to $(i,j)$ and
then never reuse $(i,j)$ in
the core grid. By doing this, we make having a monochromatic rectangle
rare and have control over when that happens.

We show how to color the cells that are not in the core grid
to achieve this.
Part I will be the final step in the reduction since we need to know
the size of the grid before we can apply it; however, we show Part I first.

Say we want the cell $(2,4)$ in the core grid to be colored $(2,4)$ and
we do not want this color appearing anywhere else in the core grid.
We can do the following: add a column of $(2,4)$'s to the left end (with one exception)
and a row of $(2,4)$'s at the bottom.
See Figure~\ref{fig:coloring24}.

\begin{figure}
\begin{large}
\[
\begin{array}{|c||c|c|c|c|c|c|c|c|c|}
\hline
(2,4) &    & &   &   &   & &   \cr
\hline
(2,4) &    & &   &   &   & &   \cr
\hline
(2,4) &   &&  &   &   & &   \cr
\hline
T     &   &(2,4) & &   &   & &   \cr
\hline
(2,4) &   &      & &   &   & &   \cr
\hline
(2,4) &   &       &   &   &   & &   \cr
\hline
(2,4) &   & &   &   &   & &   \cr
\hline
\hline
(2,4) & (2,4) &(2,4)& (2,4) & (2,4) &(2,4)  &(2,4)&(2,4) \cr
\hline
\end{array}
\]
\caption{Cell $(2,4)$ is colored $(2,4)$. No other cell can be colored $(2,4)$ in a proper coloring.}
\label{fig:coloring24}
\end{large}
\end{figure}

It is easy to see that in any extension of he 
coloring of the grid in Figure~\ref{fig:coloring24}
the only cells that can have the color $(2,4)$ are those shown
to already have that color. It is also easy to see that the color $T$ we
have will not help to create any monochromatic rectangles since
there are no other $T$'s in its column.
The $T$ we are using {\it is} the same $T$ that will later mean \emph{true}\@.  
We could have used $F$.
We do not want to use new colors since we would have no control over where else
they could be used.

What if some other cell needs to have a unique color?
Let's say we also want to color cell $(5,3)$ in the core grid with $(5,3)$
and do not want to color anything else in the core grid $(5,3)$.
Then we use the grid in Figure~\ref{fig:2453}.

\begin{figure}
\begin{large}
\[
\begin{array}{|c|c||c|c|c|c|c|c|c|c|c|}
\hline
(5,3) & (2,4) &      &      &   &   &         & & &  \cr
\hline
(5,3) & (2,4) &      &      &   &   &         & & &  \cr
\hline
(5,3) & (2,4) &      &      &  &   &          & & &  \cr
\hline
(5,3) &  T    &      &(2,4) &   &   &         & & &  \cr
\hline
T     & (2,4) &      &     &   &   & (5,3)   & & &  \cr
\hline
(5,3) & (2,4) &      &      &   &   &         & & &  \cr
\hline
(5,3) & (2,4) &      &      &       &       &       &     &     &  \cr
\hline
\hline
(5,3) & (2,4) & (2,4)& (2,4)& (2,4) & (2,4) &(2,4)  &(2,4)&(2,4)&(2,4) \cr
\hline
(5,3) & (5,3) & (5,3)& (5,3)& (5,3) & (5,3) &(5,3)  &(5,3)&(5,3)&(5,3) \cr
\hline
\end{array}
\]
\caption{$(2,4)$ and $(5,3)$ within a sub-grid}
\label{fig:2453}
\end{large}
\end{figure}

It is easy to see that in any extension of the coloring of the grid in Figure~\ref{fig:2453}
the only cells that can have the color $(2,4)$ or $(5,3)$ are those shown
to already have those colors.

For the rest of the construction we will only show the core grid.
If we denote a color as $D$ (short for {\it Distinct}) in the cell $(i,j)$
then this means that
\begin{enumerate}
\item
cell $(i,j)$ is color $(i,j)$, and
\item
we have used the above gadget to make sure that $(i,j)$ does not occur
as a color in any other cell of the core grid.
\end{enumerate}

Note that when we have $D$ in the $(2,4)$ cell and in the $(5,3)$
cell, they denote different colors.

\bigskip

\noindent
{\bf Part II: Forcing $(x,\xbar)$ to be colored $(T,F)$ or $(F,T)$.}

The first column of the core grid will have $2n$ blanks and then $m$ $D$'s.
We will use the $m$ $D$'s later.
Figure~\ref{fig:setvars} illustrates what we do in the $n=4$ case.

We will arrange things so that the color of
the blanks in Figure~\ref{fig:setvars} will all be either $T$ or $F$.
We refer to the color of the cell next to $x_i$ as {\it the color of $x_i$}.
Same for $\xibar$.

It is easy to see that in any extension of the coloring of Figure~\ref{fig:setvars}:

\begin{itemize}
\item
If $x_i$ is colored $T$ then $\xibar$ is colored $F$.
\item
If $x_i$ is colored $F$ then $\xibar$ is colored $T$.
\end{itemize}

\begin{figure}
\begin{large}
\[
\begin{array}{c|c||c|c||c|c||c|c||c|c|c|}
\hline
           & D   & D & D  & D  & D & D & D & D & D\cr
\hline
           & D   & D & D  & D  & D & D & D & D & D\cr
\hline
           & D   & D & D  & D  & D & D & D & D & D\cr
\hline
           & D   & D & D  & D  & D & D & D & D & D\cr
\hline
\hline
\xfourbar  &     & D & D  & D  & D & D & D & T & F\cr
\hline
x_4        &     & D & D  & D & D & D & D & T & F\cr
\hline
\hline
\xthreebar &     & D & D  & D & D & T & F & D & D \cr
\hline
x_3  &     & D & D  & D & D & T & F & D & D \cr
      \hline
      \hline
 \xtwobar  &     & D & D & T & F & D & D & D & D \cr
\hline
x_2        &     & D & D & T & F & D & D & D & D\cr
\hline
\hline
\xonebar   &     & T & F & D & D & D & D & D & D\cr
\hline
 x_1       &     & T & F & D & D & D & D & D & D \cr
\hline
\end{array}
\]
\caption{Literal Gadget with four variables}
\label{fig:setvars}
\end{large}
\end{figure}

We leave it to the reader to generalize Figure~\ref{fig:setvars} to $n$ variables.

We will call the left most column, which is mostly blank, {\it the literal column}.

This part is what will need to be adjusted. It will turn out that we need several
copies of each literal. During the proof of Claim 4 we will see why this is true
and how to achieve it.

\bigskip

\noindent
{\bf Part III: Forcing the coloring to satisfy a single clause}

For each clause $C=L_1 \vee L_2 \vee L_3$
we will use two columns.
These columns will be called {\it clause columns}.

Before saying what we put into the columns, Figure~\ref{fig:clausesetup} is the initial setup in the case of
$n=4$ and $m=4$. We leave it to the reader to generalize to $n,m$.
The  $X$'s in Figure~\ref{fig:clausesetup}  will be replaced by $T$'s, $F$'s, or blanks
in the next step.

\begin{figure}
\begin{large}
\[
\begin{array}{c|c|c|c|c||c|c||c|c||c|c||c|c||c|c||c|c||c|c|c|}
           &  &    &   &    &    &   &   &   &   &   & C_1 & C_1 & C_2 & C_2 & C_3 & C_3  & C_4  &C_4    \cr
\hline
\hline
           & D &    & D & D  & D  & D & D & D & D & D & D   & D   & D   & D   & D   & D    & T & T   \cr
\hline
           & D &     & D & D  & D  & D & D & D & D & D & D   & D   & D   & D   & T   & T & D & D   \cr
\hline
           & D &     & D & D  & D  & D & D & D & D & D & D & D & T & T & D & D & D & D  \cr
\hline
           & D &    & D & D  & D  & D & D & D & D & D & T & T & D & D & D & D & D & D \cr
\hline
\hline
\xfourbar  &   &  & D & D  & D  & D & D & D & T & F & X & X & X & X & X & X & X & X \cr
\hline
x_4        &   &  & D & D  & D & D & D & D & T & F & X & X & X & X & X & X & X & X \cr
\hline
\hline
\xthreebar &   &  & D & D  & D & D & T & F & D & D & X & X & X &X & X & X & X & X \cr
\hline
x_3        &   &  & D & D  & D & D & T & F & D & D & X & X & X &X & X & X & X &X \cr
 \hline
 \hline
 \xtwobar  &   &  & D & D  & T & F & D & D & D & D & X & X &X &X &X & X & X &X \cr
\hline
x_2        &   &  & D & D & T & F & D & D & D & D & X & X & X &X &X &X & X & X \cr
\hline
\hline
\xonebar   &   &  & T & F & D & D & D & D & D & D & X & X & X &X &X &X  & X & X  \cr
\hline
 x_1       &   &  & T & F & D & D & D & D & D & D & X & X & X &X &X &X & X & X \cr
\hline
\end{array}
\]
\caption{Clause setup}
\label{fig:clausesetup}
\end{large}
\end{figure}

Let $C=L_1 \vee L_2 \vee L_3$.
Figure~\ref{fig:clausegadget} illustrates how we color, or leave blank, the cells in the $C$-column.

\begin{figure}
\begin{large}
\[
\begin{array}{c||c|c|c|c|c|c|c|c|}
                &              &\cdots  &  C        & C       &\cdots \cr
\hline
\hline
                &  D           & \cdots & T         &T        & \cdots \cr
\hline
                &\vdots        &\cdots  &  \vdots   &  \vdots &\cdots  \cr
\hline
            L_3 &              &\cdots  &   D       & F       & \cdots   \cr
\hline
                &\vdots        & \cdots &  \vdots   & \vdots &\cdots  \cr
\hline
            L_2 &              &\cdots  &           &        & \cdots   \cr
\hline
                &\vdots        & \cdots & \vdots    & \vdots &\cdots  \cr
\hline
            L_1 &              &\cdots  & F         & D      & \cdots  \cr
\hline
                &\vdots        & \cdots & \vdots    & \vdots &\cdots \cr
\hline

\end{array}
\]
\caption{The clause gadget}
\label{fig:clausegadget}
\end{large}
\end{figure}

Note that since we never have the same literal appearing twice in a clause,
the construction of the Clause Gadget can be carried out.

We redraw Figure~\ref{fig:clausegadget} as Figure~\ref{fig:clausegadgetsmall}
for ease of use.  We refer to the partial coloring in Figure~\ref{fig:clausegadgetsmall}
as $\chi$.

\begin{figure}
\begin{large}
\[
\begin{array}{c||c|c|c|c|c|c|}
                &               & C & C \cr
\hline
\hline
                &  D            & T   &T\cr
\hline
                L_3 & & D  & F  \cr
\hline
                L_2            & &    &   \cr
\hline
                L_1            & & F   & D \cr
\hline
\end{array}
\]
\caption{The clause gadget---easier to work with}
\label{fig:clausegadgetsmall}
\end{large}
\end{figure}

\noindent
{\bf Claim 1:} Let $\chi$ denote the partial coloring shown in Figure 7.
If $\chi'$ is an extension of $\chi$
then $\chi'$ cannot have the $L_1,L_2,L_3$ cells all colored $F$.

\noindent
{\bf Proof of Claim 1:}

Assume, by way of contradiction,
that $L_1, L_2,L_3$ are all colored $F$.
Then we have the partial coloring in Figure~\ref{fig:allF}

\begin{figure}
\begin{large}
\[
\begin{array}{c||c|c|c|c|c|c|}
              &                & C & C \cr
\hline
\hline
              &   D            & T   &T\cr
\hline
L_3            &  F              & D  & F  \cr
\hline
L_2         &      F              &     &   \cr
\hline
L_1            &   F              & F   & D \cr
\hline
\end{array}
\]
\caption{$L_1$, $L_2$, $L_3$ all set to $F$}
\label{fig:allF}
\end{large}
\end{figure}

The reader can verify that if the two blank cells of Figure~\ref{fig:allF}
are colored $TT$, $TF$, $FT$, or $FF$, there will be a monochromatic rectangle.

\noindent
{\bf End of Proof of Claim 1}

\noindent
{\bf Claim 2:} Let $\chi'$ be an extension of the coloring in Figure~\ref{fig:clausegadgetsmall}
that colors $L_1,L_2,L_3$ but not the other two blank cells.
Assume that $\chi'$ colors $L_1, L_2, L_3$ anything except $F,F,F$.
Then $\chi'$ can be extended to color the two blank cells.

\noindent
{\bf Proof of Claim 2}

There are seven cases based on $(L_1,L_2,L_3)$ being labeled
$FFT$, 
$FTF$, 
$FTT$, 
$TFF$, 
$TFT$, 
$TTF$, 
$TTT$. 
For each one we give
a coloring of the remaining two blank cells so that
no monochromatic rectangle is formed.

\noindent
{\bf Case 1}

\[
\begin{array}{c||c|c|c|c|c|c|}
             &                & C & C \cr
\hline
\hline
             &    D            & T   &T\cr
\hline
L_3            &  F               & D  & F  \cr
\hline
L_2        &      F             &  F  &T  \cr
\hline
 L_1       &      T              & F   & D \cr
\hline
\end{array}
\]

\noindent
{\bf Case 2}

\[
\begin{array}{c||c|c|c|c|c|c|}
             &                & C & C \cr
\hline
\hline
             &    D            & T   &T\cr
\hline
L_3            &  F               & D  & F  \cr
\hline
L_2        &      T             &  T  & F \cr
\hline
 L_1       &      F              & F   & D \cr
\hline
\end{array}
\]

\noindent
{\bf Case 3}

\[
\begin{array}{c||c|c|c|c|c|c|}
             &                & C & C \cr
\hline
\hline
             &    D            & T   &T\cr
\hline
L_3            &  F               & D  & F  \cr
\hline
L_2        &      T             &   T & F \cr
\hline
 L_1       &      T              & F   & D \cr
\hline
\end{array}
\]

\noindent
{\bf Case 4}

\[
\begin{array}{c||c|c|c|c|c|c|}
             &                & C & C \cr
\hline
\hline
             &    D            & T   &T\cr
\hline
L_3            &  T               & D  & F  \cr
\hline
L_2        &      F             &   T & F \cr
\hline
 L_1       &      F              & F   & D \cr
\hline
\end{array}
\]

\noindent
{\bf Case 5}

\[
\begin{array}{c||c|c|c|c|c|c|}
             &                & C & C \cr
\hline
\hline
             &    D            & T   &T\cr
\hline
L_3            &  T               & D  & F  \cr
\hline
L_2        &      F             &   T & F \cr
\hline
 L_1       &     T               & F   & D \cr
\hline
\end{array}
\]

\noindent
{\bf Case 6}

\[
\begin{array}{c||c|c|c|c|c|c|}
             &                & C & C \cr
\hline
\hline
             &    D            & T   &T\cr
\hline
L_3            &   T              & D  & F  \cr
\hline
L_2        &      T             &  T  & F \cr
\hline
 L_1       &     F               & F   & D \cr
\hline
\end{array}
\]

\noindent
{\bf Case 7}

\[
\begin{array}{c||c|c|c|c|c|c|}
             &                & C & C \cr
\hline
\hline
             &    D            & T   &T\cr
\hline
L_3            &  T               & D  & F  \cr
\hline
L_2        &      T             &  T  & F \cr
\hline
 L_1       &      T              & F   & D \cr
\hline
\end{array}
\]

\noindent
{\bf End of Proof of Claim 2}

\bigskip

\noindent
{\bf Part IV: Putting it all together}

Recall that $\phi(x_1,\ldots,x_n)=C_1\wedge \cdots \wedge C_m$
is a 3CNF formula.
We first define the core grid and later define the
entire grid and $N,M,c$.
The core grid will have $2n+m$ rows and $2m+2n+1$ columns
(when we later modify the construction the core grid will be bigger though still linear in $n,m$)
The $2n$ left-most columns are partially colored, and labeled with literals,  as described in Part II.
The $m$ top-most rows are colored, and labeled with clauses,  as described in Part III.
The rest of the core grid is colored as described in Part III.

The core grid is now complete.
For every $(i,j)$ that is colored $(i,j)$, we perform the method in Part I
to make sure that $(i,j)$ is the only cell with color $(i,j)$.
Let the number of such $(i,j)$ be $E$.
The number of colors $c$ is $E+2$. This will force everything else to be colored $T$ or $F$.
Note that $E=\Theta(NM)$.

In Figure~\ref{fig:whyfails} we present the core grid for the instance of $\GCE$ obtained if the
original formula is

$$
(x_1\vee x_2 \vee \xthreebar)\wedge(x_1\vee x_2\vee x_4)\vee(\xtwobar\vee x_3\vee x_4).
$$

\begin{figure}[h]
\begin{large}
\[
\begin{array}{c|c|c|c|c||c|c||c|c||c|c||c|c||c|c||c|c|}
           &  &    &   &    &    &   &   &   &   &   & C_1 & C_1 & C_2 & C_2 & C_3 & C_3   \cr
\hline
\hline
           & D &    & D & D  & D  & D & D & D & D & D & D   & D   & D   & D   & T   & T   \cr
\hline
           & D &     & D & D  & D  & D & D & D & D & D & D & D & T & T & D & D  \cr
\hline
           & D &    & D & D  & D  & D & D & D & D & D & T & T & D & D & D & D \cr
\hline
\hline
\xfourbar  & &  & D & D  & D  & D & D & D & T & F & D & D & D & D & D & D  \cr
\hline
x_4        & &  & D & D  & D  & D & D & D & T & F & D & D & D & F & D & F \cr
\hline
\hline
\xthreebar & &  & D & D  & D & D & T & F & D & D & D & F & D &D & D & D \cr
\hline
x_3        & &  & D & D  & D & D & T & F & D & D & D & D & D &D & &  \cr
 \hline
 \hline
 \xtwobar  & &  & D & D  & T & F & D & D & D & D & D & D &D   &D &F &D \cr
\hline
x_2        & &  & D & D & T & F & D & D & D & D &  & &  & &D &D  \cr
\hline
\hline
\xonebar   & &  & T & F & D & D & D & D & D & D &   D & D & D &D &D &D  \cr
\hline
 x_1       & &  & T & F & D & D & D & D & D & D &   F & D & F &D &D &D \cr
\hline
\end{array}
\]
\caption{Example with $(x_1\vee x_2 \vee \xthreebar) \wedge (x_1\vee x_2 \vee x_4) \wedge (\xtwobar\vee x_3 \vee x_4) $}
\label{fig:whyfails}
\end{large}
\end{figure}

\noindent
{\bf Claim 3:}
Let $\phi(x_1,\ldots,x_n)$ be a 3CNF formula.
Let $(N,M,c,\chi)$ be the result of the reduction described above.
If $(N,M,c,\chi)\in \GCE$ then $\phi\in\thSAT$.

\noindent
{\bf Proof of Claim 3}

\noindent
Assume that  $(N,M,c,\chi)\in \GCE$. 
According to the construction in Part II
the first column gives a valid truth
assignment for 
$x_1,\ldots,x_n$ (and hence also for
$\xonebar, \ldots, \xnbar$). By Claim 1, for every clause $C=L_1\vee L_2\vee L_3$
this truth assignment cannot assign $L_1,L_2$ and $L_3$ all to $F$.
Hence this is a satisfying assignment, so $\phi\in\thSAT$.

\noindent
{\bf End of Proof of Claim 3}

We will now try to show that 
if $\phi\in\thSAT$ then $(N,M,c,\chi)\in \GCE$.
{\bf We will fail!} This will motivate us to modify our construction. 

\noindent
{\bf Claim 4 (which is false):}
Let $\phi(x_1,\ldots,x_n)$ be a 3CNF formula.
Let $(N,M,c,\chi)$ be the result of the reduction described above.
If $\phi\in\thSAT$ then $(N,M,c,\chi)\in \GCE$.

\noindent
{\bf Proof of Claim 4 (which will fail)}

Assume $\phi\in\thSAT$. Let $(b_1,\ldots,b_n)$ be a satisfying truth assignment
where, for $1\le i\le n$, $b_i\in \{T,F\}$.
We use this to obtain a coloring of $\GNM$ that is an extension of $\chi$.

Color the literal column in the obvious way: the entry labeled with literal $L$
is labeled the truth assignment of $L$.
We now show how we try to color the blank cells in the clause columns.

Let $C=L_1\vee L_2\vee L_3$ be a clause.
The part of the grid associated to it is in Figure~\ref{fig:clausegadget}.

The literal column we have already colored. Since the assignment was satisfying,
at least one of $L_1,L_2,L_3$ was set to $T$.
We use Claim 2 to extend the coloring to the blank cells.
This forms a grid coloring. We try to prove this coloring is proper.

Assume, by way of contradiction, that there is a monochromatic rectangle
which we call $R$.

\noindent
{\bf Case 1} There is a clause $C$ such that $R$ uses the two $T$'s associated to $C$.
The only way these $T$'s can be involved in a monochromatic rectangle is if the two blank cells
associated to $C$ are colored $T$. By the 7 cases in Claim 2 this cannot occur.

\noindent
{\bf Case 2}
There is a variable $x$ such that $R$ uses the two $T$'s or two $F$'s associated to $x$.
Figure~\ref{fig:case2} shows what this looks like (we only include the relevant parts).
We assume $x$ is the first variable in $C$ (the other cases are either similar or cannot occur). 

\begin{figure}
\begin{large}
\[
\begin{array}{c|c||c|c||c|c||c|}
              & \cdots   & D & D  & \cdots  & C & \cdots \cr
\hline
              & \vdots   & \vdots & \vdots  & \vdots  & D & \vdots \cr
\hline
\xbar         & \cdots     & T & F & \cdots & F & \cdots \cr
\hline
x             &  \cdots    & T & F &\cdots  & F & \cdots \cr
\end{array}
\]
\caption{Case 2 of Claim 4}
\label{fig:case2}
\end{large}
\end{figure}

No clause-column has two $T$'s in it, so $R$ must be colored $F$.
The only way there can be two $F$'s in the literal-column is if they
are associated to a literal and its negation, as in Figure~\ref{fig:case2}.
However, the only way that configuration can happen is if $x$ and $\xbar$
are in the same clause.
This cannot happen since $\phi$ has no clauses with both a variable and its negation in it.

\noindent
{\bf Case 3} $R$ uses the literal column and one of the clause columns.
By Claim 2, $R$ is not monochromatic.

\noindent
{\bf Case 4} The only case left is if $R$ uses two clause columns. {\it This can occur!}
This is where the construction fails!
We give an example. Recall that Figure~\ref{fig:whyfails} is the instance of $\GCE$ from
the formula
$$(x_1\vee x_2 \vee \xthreebar) \wedge (x_1\vee x_2 \vee x_4) \wedge (\xtwobar\vee x_3 \vee x_4).$$
Lets say we take the satisfying truth assignment
$$x_1=T, x_2=F, x_3=T, x_4=F.$$

If we put these in the literal column
and use the proof of Claim 2 to color the blank cells in the clause columns,
the result is the coloring of the entire grid seen in Figure~\ref{fig:fails}.
The boldfaced colors are the ones caused by the truth assignment. 
The asterisks show a monochromatic rectangle. Hence the construction produces a
non-proper coloring and is incorrect. 

\begin{figure}
\begin{large}
\[
\begin{array}{c|c|c|c|c||c|c||c|c||c|c||c|c||c|c||c|c|}
           &  &    &   &    &    &   &   &   &   &   & C_1 & C_1 & C_2 & C_2 & C_3 & C_3   \cr
\hline
\hline
           & D &     & D & D  & D  & D & D & D & D & D & D   & D   & D   & D   & T   & T  \cr
\hline
           & D &     & D & D  & D  & D & D & D & D & D & D & D & T & T & D & D  \cr
\hline
           & D &    & D & D  & D  & D & D & D & D & D & T & T & D & D & D & D \cr
\hline
\hline
\xfourbar  & \bm{T} &  & D & D  & D  & D & D & D & T & F & D & D & D & D & D & D  \cr
\hline
x_4        & \bm{F} &  & D & D  & D & D & D & D & T & F & D & D & D & F & D & F \cr
\hline
\hline
\xthreebar & \bm{F} &  & D & D  & D & D & T & F & D & D & D & F & D &D & D & D \cr
\hline
x_3        & \bm{T} &  & D & D  & D & D & T & F & D & D & D & D & D  &D &\bm{T}  &\bm{F}  \cr
 \hline
 \hline
 \xtwobar  & \bm{T} &  & D & D  & T & F & D & D & D & D & D & D &D   &D &F &D \cr
\hline
x_2        & \bm{F} &  & D & D & T & F & D & D & D & D &  \bm{F}* & \bm{T}  & \bm{F}*  & \bm{T}&D &D  \cr
\hline
\hline
\xonebar   & \bm{F} &  & T & F & D & D & D & D & D & D &   D & D & D &D &D &D  \cr
\hline
 x_1       & \bm{T} &  & T & F & D & D & D & D & D & D &   F* & D & F* &D &D &D \cr
\hline
\end{array}
\]
\caption{Example with $(x_1\vee x_2 \vee \xthreebar) \wedge (x_1\vee x_2 \vee x_4) \wedge (\xtwobar\vee x_3 \vee x_4) $}
\label{fig:fails}
\end{large}
\end{figure}

\noindent
{\bf End of the Proof of Claim 4 (that failed)}

The way to avoid Case 4 is if we had {\it several}
copies of each literal so that if two clauses use the same literal,
they will use different copies of it.
How many? The number of copies of literal $L$ has to be at least the
number of clauses that $L$ appears in. It will be convenient to have
the number of copies of $L$ and of $\overline{L}$ be the same.
Hence if $x$ appears in $m_1$ clauses, and $\xbar$ appears in $m_2$ clauses, 
then we'll add $\max\{m_1,m_2\}$ rows for each of these literals.

Rather than give the general construction, we do an example with
the case that gave us trouble before:
$$
(x_1\vee x_2 \vee \xthreebar)\wedge(x_1\vee x_2 \vee x_4) \wedge(\xtwobar\vee x_3\vee x_4).
$$
in Figure~\ref{fig:works}.
We prove that the literals $x_1$ and $\xonebar$ behave as they should. From the example and
the proof about $x_1,\xonebar$, the reader will be able to work out the general construction.

\bigskip

\noindent
{\bf Claim 5} Let $\chi'$ be a proper extension of the coloring in Figure~\ref{fig:works}
that colors all entries of the literal column.
Then the following three statements are true.
\begin{enumerate}
\item
The $x_1$ cells all have the same color: $T$ or $F$.
\item
The $\xonebar$ cells all having the same color: $T$ or $F$.
\item
The color of $x_1$ and $\xonebar$ are different.
\end{enumerate}

\noindent
{\bf Proof of Claim 5}

We coordinate the grid by having the bottom left cell be $(1,1)$ (which is blank), 
the cell to the right has
coordinate $(2,1)$ (which has $T$), and the cell above it is $(1,2)$ (which is blank).

Assume $\chi'(1,1)=T$ (the case where $\chi'(1,1)=F$ is similar).

Since $\chi'(1,1)=T$, $\chi'(2,1)=T$, and  $\chi'(2,2)=T$, we have $\chi'(1,2)=F$.

Since $\chi'(1,2)=F$, $\chi'(5,2)=F$, and  $\chi'(5,3)=F$, we have $\chi'(1,3)=T$.

Since $\chi'(1,3)=T$, $\chi'(6,3)=T$, and  $\chi'(6,4)=T$, we have $\chi'(1,4)=F$.

\noindent
{\bf End of Proof of Claim 5}

From this the reader can work out the general construction.
\begin{figure}
\begin{large}
\[
\begin{array}{c||c||c|c|c|c|c|c||c|c|c|c|c|c||c|c||c|c||c|c||c|c||c|c|} 
          & &   &   &   &  &   &   &  &       &  &  &   &  &   &     &  &  &C_1  &C_1  &C_2  &C_2  &C_3 &C_3     \cr
\hline
          &D&D  &D  &D  &D & D &D & D&D & D   &D & D & D &D &D   &D &D & D&D &D &D &T&T  \cr
\hline
          &D&D  &D  &D  &D & D &D & D&D & D   &D & D & D &D &D   &D &D & D&D &T &T &D&D  \cr
\hline
          &D&D  &D  &D  &D & D &D & D&D & D   &D & D & D &D &D   &D &D &T &T &D &D &D&D  \cr
\hline
\hline
\xfourbar  & & D & D &D  &D &D  &D &D &D &D    &D &D  &D &D  &D    &T &F &D &D &D &D &D &D \cr
\hline
x_4        & & D & D &D  &D &D  &D &D &D &D    &D &D  &D &D  &D    &T &F &D &D &D &F &D &F \cr
\hline
\hline
\xthreebar& & D & D &D  &D &D  &D  &D &D &D & D   &D &D  &T  &F    &D &D &D &F &D &F &D &D \cr
\hline
x_3      & & D & D &D  &D &D  &D   &D &D &D & D    &D &D &T  &F    &D &D &D &D &D &D &  &  \cr
\hline
\hline
\xtwobar & & D & D &D  &D &D  &D   &D &D &D  &D    &T &F &D  &D     &D &D &D &D &D &D &D &D  \cr
\hline
x_2      & & D & D &D  &D &D  &D   &D & D&T & F    &T &F &D  &D     &D &D &D &D &  &  &D &D  \cr
\hline
\xtwobar & & D & D &D  &D &D  &D   &T &F &T  &F    &D &D &D  &D     &D &D &D &D &D &D &F &D  \cr
\hline
x_2      & & D & D &D  &D &D  &D   & T&F & D & D   &D &D &D  &D     & D& D&  &  &D &D &D &D  \cr
\hline
\hline
\xonebar & & D & D &D  &D &T  &F   &D &D &D & D    &D &D &D &D   &D &D  &D &D &D &D &D &D \cr
\hline
x_1      & & D & D &T  &F &T  &F   &D &D &D & D    &D &D &D &D   &D &D &D &D &F &D &D &D \cr
\hline
\xonebar & & T & F &T  &F &D  &D   &D &D &D & D    &D &D &D &D   &D &D &D &D &D &D &D &D \cr
\hline
x_1      & & T & F &D  &D &D  &D   &D &D &D & D    &D &D &D &D   &D &D &F &D &D &D &D &D  \cr
\hline
\end{array}
\]
\caption{Example with $(x_1\vee x_2 \vee \xthreebar) \wedge (x_1\vee x_2 \vee x_4)\wedge (\xtwobar \vee x_3\vee x_4)$}
\label{fig:works}
\end{large}
\end{figure}
\end{proof}

\noindent
{\bf Recap and the Actual Values of $N,M,c$}

Our goal was to, given 
a 3CNF formula
$$\phi(x_1,\ldots,x_n)=C_1\wedge \cdots\wedge C_m,$$
with $n$ free variables and $m$ clauses,
output an instance of $\GCE$ such that 

$$\phi\in \thSAT\hbox{ iff } (N,M,c,\chi)\in \GCE.$$

We have described the partial coloring $\chi$.
For the sake of completeness we now specify $N,M,c$. 
We first find the dimensions of the core grid. 

\begin{definition}
$o_i$ is the maximum of the number of occurrence of $x_i$ and $\xibar$. 
\end{definition}

\begin{lemma}\label{le:bsum}
$\sum_{i=1}^n o_i \le 3m$. 
 \end{lemma}

\begin{proof}
Since $o_i$ is the number of times one of $\{x_i,\xibar\}$ occurs
the sum is bounded by the number of occurrence of variables.
Since the formula is in 3CNF form, the number of occurrence of
variables is $3m$. 
\end{proof}

We use Lemma~\ref{le:bsum} to get upper bounds on several quantities
including $N,M,c$. 

\bigskip

\noindent
{\bf Number of rows in the core grid}
The literal column will have $o_i$ rows labeled $x_i$ and
$o_i$ rows labeled $\xibar$. Hence the literal column will have
$\sum_{i=1}^n 2o_i$ blank cells. 
Every clause $C$ induces a row.
(The row has all $D$'s except for two $T$'s under the columns labeled $C$;
however, we do no need that to count the number of rows.)
Hence there are $m$ additional rows.
Therefore the number of rows in the core grid is

$$N'=m+\sum_{i=1}^n 2o_i=m+2\sum_{i=1}^n o_i \le m + 3m=3m.$$

\noindent
{\bf Number of columns in the core grid}
Each variable $x_i$ induces a rectangle of height $2o_i$
and width $4o_i-2$. See Figure~\ref{fig:whyfails} for an example with 
$o_1=1$, Figure~\ref{fig:works} for an example with $o_1=2$, and
Figure~\ref{fig:threeocc} for an example with $o_1=3$. 
Each clause adds 2 columns. 
Therefore the number of columns in the core grid is

$$M'=2m+\sum_{i=1}^n (4o_i-2) = 2m + 4\sum_{i=1}^n o_i\le 2m+4\times 3m=14m\le 14m.$$

Note that $N',M'$ are linear in $n,m$ as promised earlier. 
However, the non-core part of the grid will add an $O(m^2)$ term to the size of $N,M$.

\begin{figure}
\begin{large}
\[
\begin{array}{c||c||c|c|c|c|c|c|c|c|c|c|} 
\hline
\xonebar & & D & D &D  &D &D  &D   &T &F &T  &F     \cr
\hline
x_1      & & D & D &D  &D &D  &D   &T &F & D & D    \cr
\hline
\xonebar & & D & D &D  &D &T  &F   &T &F &D & D     \cr
\hline
x_1      & & D & D &T  &F &T  &F   &D &D &D & D    \cr
\hline
\xonebar & & T & F &T  &F &D  &D   &D &D &D & D     \cr
\hline
x_1      & & T & F &D  &D &D  &D   &D &D &D & D     \cr
\hline
\end{array}
\]
\caption{Three Occurrence of $x_1$}
\label{fig:threeocc}
\end{large}
\end{figure}

\noindent
{\bf The Number of Blank Cells, $T$-Cells, $F$-Cells, and Colors} 

The first column has $\sum_{i=1}^n 2o_i$ blank cells.
Each column labeled with a clause has 1 blank cell.
Hence the number of blank cells is

$$B=m+\sum_{i=1}^n 2o_i= m+2\sum_{i=1}^n o_i.$$

Each column that is not labeled with a clause has one $T$ and one $F$.
Each column labeled with a clause has one $T$ and one $F$. Hence the number
of cells labeled with a $T$ or an $F$ is

$$2M'$$

Every cell that is neither blank, $T$, or $F$ has a distinct color.
Hence the number of new colors that are not $T$ or $F$ is 

$$E=N'M'-B-2M'\le N'M' \le 42m^2.$$

and the total number of colors is

$$c=E+2\le 42m^2+2.$$

\noindent
{\bf The real values of $N,M$}

We now deal with the non-core part of the grid. For every color that is not $T$ or $F$
we add one row and one column to the grid (see Part I of the construction). Hence

$$M=M'+E \le 14m + 42m^2 = O(m^2).$$

$$N=N'+E \le 3m + 42m^2 = O(m^2).$$

Note that $M,N$ are polynomial in the length of $\phi$.

\section{What the $\NP$-Completeness Result Does and Does Not Tell Us}\label{se:no}

The motivation for this paper was

{\it Why was finding if $G_{17,17}$ is 4-colorable so hard?}

Towards this goal we showed, in Theorem~\ref{th:npc}, that $\GCE$ is $\NP$-complete.
But does this really capture the problem we want to study?
We give several reasons why not. These will point to further investigations.

\noindent
1) The reduction in Theorem~\ref{th:npc} takes a 3CNF formula

$$\phi(x_1,\ldots,x_n) = C_1\wedge \cdots \wedge C_m$$

\noindent
and produces an instance $(N,M,c,\chi)$ of $\GCE$ such that

$$\phi\in \thSAT \hbox{ iff } (N,M,c,\phi)\in\GCE.$$

\noindent
In this instance $c=\Theta(NM$).
Hence our reduction only shows that $\GCE$ is hard if
$c$ is rather large.
So what happens if $c$ is small? See next point.

\bigskip

\noindent
2) What happens if $c$ is small? In Section~\ref{se:fpt} we show that
$\GCE$ is Fixed Parameter Tractable. In particular, the problem is in
time $O(N^2M^2) + 2^{O(c^4+\log c)}$. This leads to the following open problem:
find a framework to show that some problems in FPT are hard.

\bigskip

\noindent
3) The $17\times 17$ challenge can be rephrased as proving that
$(17,17,4,\chi)\in \GCE$ where $\chi$ is the empty partial coloring.
This is a special case of $\GCE$ since none of the
cell are pre-colored.
It is possible that the case where $\chi$ is the empty coloring
is easy.
While we doubt this is true, we have not eliminated the possibility.
How to deal with this issue?
We define the problem that is probably the one we really want to find the complexity of.

\begin{definition}
$\GC$ is the following problem:
\begin{itemize}
\item
{\it Input}  $M,N,c\in\N$. 
The numbers $N,M,c$ are in unary. So formally the input is
$(1^M,1^N,1^c)$ where $1^x$ means $1\cdots 1$ ($x$ times). 
\item
{\it Output} YES if there is a total $c$-coloring of 
$\GNM$, NO otherwise.
\end{itemize}
$\GC$ stands for {\it Grid Coloring}.
\end{definition}

Clearly $\GC\in\NP$. 
Is this problem $\NP$-hard? Alas no (assuming $\P\ne\NP$). 

\begin{definition}
A set $X\subseteq \{0,1\}^*$ is {\it sparse} if
there exists a polynomial $p$ such that 

$$(\forall n)[|X \cap \{0,1\}^n | \le p(n)]$$
\end{definition}

Note that $\GC\subseteq 1^*\times 1^* \times 1^*$ and hence is a sparse set. 

\noindent
We state a theorem that indicates sparse sets are not $\NP$-hard. 

\begin{theorem}~
\begin{enumerate}
\item 
(Mahaney~\cite{Mahaney-1982}, see also \cite{HL-1994} for an alternative proof). 
If there exists a sparse set that is $\NP$-hard by an $m$-reduction then $\P=\NP$. 
\item
(Karp-Lipton Theorem~\cite{KL-1980})
If there exists a sparse set that is $\NP$-hard by a Turing-reductions then $\Sigma_2^p=\Pi_2^p$. 
\end{enumerate}
\end{theorem} 

Hence $\GC$ is likely to not be $\NP$-complete under either $m$-reductions or Turing-Reductions. 

If in the $\GC$ problem we express $N,M$ in binary, then we cannot show that $\GC$ is in $\NP$ 
since the obvious witness, the coloring,  is exponential in the length of the input. 
The formulation in binary does not get at 
the heart of the problem, since we believe it is hard because the number of possible colorings is 
large, not because $N,M$ are large.

\section{Fixed Parameter Tractability}\label{se:fpt}

Consider the problem where the number of colors is fixed at some $c$.
We will see that this problem is Fixed Parameter Tractable by presenting two
FPT algorithms for it.

The algorithms we present are not only FPT; they achieve this by means of a
{\it polynomial kernel}. We discuss this in a subsection after the algorithms.

How well does the algorithm do in practice? 
We discuss this in a second subsection after the algorithm. 
In particular we discuss how much time and space the algorithm takes on one of our 
motivating problems:
determining if there is a 4-coloring of $G_{17,17}$. 
The punchline will be that the algorithm takes too much time, and too much space, to be practical.
Even so, we present the algorithm in the hope that some clever reader can come up
with a way around these limitations, perhaps in practice if not in theory.

\begin{definition}
Let $c\in\N$.
$\GCE_c$ is the following problem:
\begin{itemize}
\item
{\it Input}  $N,M\ge 1$ and $\chi$ a partial $c$-coloring of $\GNM$.
The numbers $N,M$ are in unary. 
\item
{\it Output} YES if there is an extension of $\chi$ to a total $c$-coloring of 
$\GNM$, NO otherwise.
(The algorithm can easily be modified to also {\it output} the extension as well as the YES.) 
\end{itemize}
\end{definition}

Clearly $\GCE_c \in \DTIME(c^{O(NM)})$.
We will show that $\GCE_c$ is in time $O(N^2M^2)+2^{O(c^6+\log c)}$ and
then improve the algorithm to show that $\GCE_c$ is in time
$O(N^2M^2)+2^{O(c^4+\log c)}$.

\begin{lemma}\label{le:comb}
For all $u\ge 0$, $\sum_{s=0}^u \binom{u}{s}2^s = 3^u$.
\end{lemma}

\begin{proof}

$$3^u = (1+2)^u = \sum_{s=0}^u \binom{u}{s}2^s.$$

The last equality is by the binomial theorem. 
\end{proof}

\begin{lemma}\label{le:ext}~
Assume that $\GNM$ is partially $c$-colored by $\chi$.
Let $S$ be a set of cells that are not colored by $\chi$.
Let $|S|=s$.
Let $\chi^*$ be a (not necessarily proper) extension of $\chi$ that colors
all of the cells of $S$. We can determine whether $\chi^*$ is a proper $c$-coloring
in time $O(NMs)$.
\end{lemma}

\begin{proof}

\noindent
Here is the algorithm.

\noindent
\begin{enumerate}
\item[ ]
For each $(x,y)\in S$ do the following.
\begin{enumerate}
\item[ ]
For each $1\le x'\le N$ and $1\le y'\le M$ determine if

$$\chi^*(x,y)=\chi^*(x',y)=\chi^*(x,y')=\chi^*(x',y').$$

\noindent
If the equality ever holds then {\it output} NO and stop.
\end{enumerate}

\item[ ]
(If you get here then the equality never happened.)
{\it Output} YES and stop.
\end{enumerate}

The first for-loop goes $s$ iterations. The second for-loop
goes $NM$ iterations. The body of the for-loop is $O(1)$ time. Hence
the run time is $O(NMs)$.
\end{proof}

\begin{lemma}\label{le:dyn}
Let $N,M,c\in\N$.
Let $\chi$ be a partial $c$-coloring of $\GNM$.
Let $U$ be the uncolored grid cells. Let $|U|=u$.
There is an algorithm that takes $O(cuNM3^u)$ time and $2^uc$ space 
that will determine if $\chi$ can
be extended to a full $c$-coloring.
\end{lemma}

\begin{proof}
For $S\subseteq U$ and $0\le i\le c$ let
\begin{equation*}
f(S,i)=
\begin{cases}
\YES & \text{ if $\chi$ can be extended to color $S$ using only colors $\{1,\ldots,i\}$; } \\
\NO  & \text{ if not.} \\
\end{cases}
\end{equation*}

We assume throughout that the coloring $\chi$ has already been applied.

We are interested in $f(U,c)$; however, we use a dynamic program
to compute $f(S,i)$ for all $S\subseteq U$ and $0\le i\le c$.
Note the base cases:

\begin{enumerate}
\item
$f(\es,i)=\YES$.
\item
If $S\ne\es$ then $f(S,0)=\NO$.
\end{enumerate}

\noindent
{\bf Claim 1} Let $S\subseteq U$ and $1\le i\le c$.
Assume that, for all $S'$ such that $|S'|<|S|$,  for all $0\le i\le c$, $f(S',i)$ is known.
Also assume that $f(S,i-1)$ is known.
Let $|S|=s$. Then $f(S,i)$ can be determined in time $O(NMs2^s)$.

\noindent
{\bf Proof of Claim 1}

If $f(S,i-1)=\YES$ then clearly $f(S,i)=\YES$. If not then
here is our plan: We want to find (or show there is no such)
nonempty $T\subseteq S$ such that the following holds:
\begin{itemize}
\item
$f(S-T,i-1)=\YES$. 
Hence there is a proper extension of $\chi$ which uses colors $\{1,\ldots,i-1\}$ on $S-T$. 
Let $\chi^*$ be that coloring.
Note that, for all nonempty $T\subseteq S$ the value $f(S-T,i-1)$ is known since $|S-T|<|S|$.
\item
The extension of $\chi$ obtained by coloring all cells in $T$ with $i$ is a proper coloring.
\end{itemize}

If we find such a $T$ then clearly $f(S,i)=\YES$ by using $\chi^*$ to color $S-T$
with $\{1,\ldots,i-1\}$ and
then coloring all cells in $T$ with $i$.
The algorithm below tries to find such a $T$. It will be clear that if the
algorithm says $\YES$ then there is such a $T$ and hence $f(S,i)=\YES$.
We will need to prove that if the algorithm says $\NO$ then $f(S,i)=\NO$.

\begin{enumerate}
\item
If $f(S,i-1)=\YES$ then {\it output} $\YES$ and stop.
\item
For all nonempty $T\subseteq S$ do the following (Note that there are $2^{s}-1$ nonempty sets $T$.)
\begin{enumerate}
\item
Let $\chi'$ be the extension of $\chi$ that colors all cells in $T$ with $i$. 
\item
If $\chi'$ is not a proper coloring then go to the next $T$.
Note that, by Lemma~\ref{le:ext}, this takes $O(NM|T|)=O(NMs)$ time. 
\item
If $f(S-T,i-1)=\NO$ then go to the next $T$.
Note that we know the value of $f(S-T,i-1)$ because $|S-T|<|S|$.
This step takes $O(1)$ time. 
\item
If the algorithm got to this step then the following have happened:
\begin{enumerate}
\item 
$\chi'$ is proper.
\item
$f(S-T,i-1)=\YES$. 
Hence there is a proper extension of $\chi$ which uses colors $\{1,\ldots,i-1\}$ on $S-T$. 
Let $\chi^*$ be that coloring.
\end{enumerate}
The extensions $\chi'$ and $\chi^*$ can easily be combined to properly extend $\chi$
to $S$ with colors $\{1,\ldots,i\}$. Hence $f(S,i)=\YES$ and we have the coloring itself. 
\end{enumerate} 
\item
If the algorithm got to this step then no $T$ worked. We will show that in this case
$f(S,i)=\NO$.
\end{enumerate}

The algorithm just specified has $2^s$ iterations that take $O(NMs)$ each.
Hence the algorithm runs in time $O(NMs2^s)$.

Clearly if the above algorithm outputs $\YES$ then $f(S,i)=\YES$.
We need to show if the output is $\NO$ then $f(S,i)=\NO$.

\noindent
{\bf Claim 2:} If the above algorithm outputs $\NO$ then $f(S,i)=\NO$.

\noindent
{\bf Proof of Claim 2:}
If the above algorithm outputs $\NO$ then, for all nonempty
$T\subseteq S$ at least one of the following cases holds:

\noindent
{\bf Case 1:}
The extension of $\chi$ to $T$ formed by coloring cells of $T$ with $i$ is not proper.

\noindent
{\bf Case 2:} $f(S-T,i-1)=\NO$.

Assume, by way of contradiction, that $f(S,i)=\YES$.
Let $\COL$ be a proper extension of $\chi$ to $S$.
Let $T$ be the subset of $S$ that is colored $i$.

Since $\COL$ is a proper extension of $\chi$,
the extension of $\chi$ to $T$ formed by coloring cells in $T$ with $i$ is proper.
So Case 1 does not apply to $T$.

Since $\COL$ is a proper extension of $\chi$ there is a proper extension of $\chi$
to $S-T$ that only uses $\{1,\ldots,i-1\}$.
So Case 2 does not apply to $T$.

Neither case applies, which is a contradiction.

\noindent
{\bf End of Proof of Claim 2}

\noindent
{\bf End of Proof of Claim 1}

We use Claim 1 in the following dynamic program.

\begin{enumerate}
\item
{\it Input$(M,N,c,\chi)$} such that $\chi$ is a partial $c$-coloring of $\GNM$.
Let $U$ be the set of cells that are not colored by $\chi$.
Let $|U|=u$.
\item
Set up a 2 dimensional table indexed by $S\subseteq U$ and $0\le i\le c$.
\item
Set $f(\es,i)=\YES$.
\item
If $S\ne\es$ then set $f(S,0)=\NO$.
\item
For $S\subseteq U$ (go in order of size)
\begin{enumerate}
\item[ ]
For $i=1$ to $c$ determine $f(S,i)$ using Claim 1 which takes time $O(NMs2^s)$.
\end{enumerate}
\end{enumerate}

Note that the amount of time taken in the inner loop, $O(NMs2^s)$ is
independent of $c$. That is why $c$ will only appear linearly in the run time.

The number of subsets of $U$ that have $s$ cells is $\binom{u}{s}$. Hence the
total time spent in the loop is O-of the following:

$$\sum_{i=1}^c \sum_{s=0}^u \binom{u}{s} sNM2^s \le  cuNM\sum_{s=0}^u \binom{u}{s}2^s$$

By Lemma~\ref{le:comb}, $\sum_{s=0}^u \binom{u}{s}2^s=3^u$, so we obtain $O(cuNM3^u).$
\end{proof}

The following two theorems are easy; however, we include the proofs for completeness. 

\begin{lemma}\label{le:bounds}
Assume $c+1\le N$ and  $c\binom{c+1}{2} < M$.
Then $\GNM$ is not $c$-colorable.
Hence, for any $\chi$, $(N,M,\chi)\notin \GCE_c$.
\end{lemma}

\begin{proof}
Assume, by way of contradiction, that there is a $c$-coloring of $\GNM$.
Since every column has at least $c+1$ cells, each column must have two cells that have the
same color. 
Map every column to some $(\{i,j\},a)$ such that the $i$th and the $j$th entry in that column
are both colored $a$.
Since the number of $(\{i,j\},a)$ is 

$$\binom{N}{2}\times c \le \binom{c+1}{2}\times c< M,$$

\noindent
two columns must map to the same $(\{i,j\},a)$. This will create a monochromatic
rectangle, which is a contradiction. 
\end{proof}

\begin{lemma}\label{le:c}
Assume $N\le c$ and $M\in \N$. If $\chi$ is a partial $c$-coloring of $\GNM$
then $(N,M,\chi)\in \GCE_c$.
\end{lemma}

\begin{proof}
The partial $c$-coloring $\chi$ can be extended to a full $c$-coloring as follows:
for each column  use a different color for
each blank cell, making sure that all of the new colors
in that column are
different from each other and from the already existing colors given by $\chi$. 
\end{proof}

\begin{theorem}
$\GCE_c$ can be computed in time $O(N^2M^2)+2^{O(c^6+\log c)}$ and
space $O(c2^{c^6})$. 
\end{theorem}

\begin{proof}~

\begin{enumerate}
\item
{\it Input $(N,M,\chi)$}.

\item
If $N\le c$ or $M\le c$ then test if $\chi$ is
a partial $c$-coloring of $\GNM$. If so then {\it output} $\YES$. 
If not then {\it output} $\NO$.
(This works by Lemma~\ref{le:c}.)
This takes time $O(N^2M^2)$.
Henceforth we assume $c+1\le N,M$.

\item
If  $c\binom{c+1}{2} < M$ or $c\binom{c+1}{2} < N$ then {\it output} $\NO$ and stop.
(This works by Lemma~\ref{le:bounds}.)

\item
The only case left is $c+1\le N,M\le c\binom{c+1}{2}$.
We will apply Lemma~\ref{le:dyn}. Note that the number of uncolored cells,
$u$, is 

$$\le NM\le  (c\binom{c+1}{2})^2 \le (c \times \frac{(c+1)^2}{2})^2 = O(c^6).$$

Hence the run time of this step is

$$O(cuNM3^u)=O(cc^6c^6 3^{c^6}) = 2^{O(c^6+\log c)}.$$

\end{enumerate}

Step 2 takes $O(N^2M^2)$, and Step 4 takes time $2^{O(c^6+\log c)}$. Hence the
entire algorithm takes time $O(N^2M^2) + 2^{O(c^6+\log c)}$.
\end{proof}

Can we do better? Yes, but it will require a result from a paper
by Fenner et al. \cite[Corollary 2.12]{grid}.

\begin{lemma}\label{le:better}
Let $1\le c'\le c-1$.
\begin{enumerate}
\item
If $N\ge c+c'$ and $M > \frac{c}{c'}\binom{c+c'}{2}$ then $\GNM$ is not $c$-colorable.
\item
If $N\ge 2c$ and $M>2\binom{2c}{2}$ then $\GNM$ is not $c$-colorable.
(This follows from a weak version of the $c'=c-1$ case of Part I.)
\end{enumerate}
\end{lemma}

\begin{theorem}\label{th:fpt4}
$\GCE_c$ can be computed in time $O(N^2M^2)+2^{O(c^4+\log c)}$ and 
$O(c2^{NM})$ space. 
\end{theorem}

\begin{proof}~

\begin{enumerate}

\item
{\it Input $(N,M,\chi)$}.  Let $u=NM$ which is a bound on the number of cells that are not
colored.

\item
If $N\le c$ or $M\le c$ then test if $\chi$ is
a partial $c$-coloring of $\GNM$. If so then {\it output} $\YES$. 
If not then {\it output} $\NO$.
(This works by Lemma~\ref{le:c}.)
This takes time $O(N^2M^2)$.

\item
Let $c'=N-c$ and $c''=M-c$. 

\item
If $c'\le c-1$ then do the following. Note that $N=c+c'$ and $M=c+c''$. 
\begin{enumerate}
\item
If $M> \frac{c}{c'}\binom{c+c'}{2}$, then {\it output} $\NO$ and stop.
(This works by Lemma~\ref{le:better}.)
\item
If $M\le \frac{c}{c'}\binom{c+c'}{2}$ then do the following.
By Lemma~\ref{le:dyn}
we can determine if $\chi$ can be extended to a total $c$-coloring in time
$O(cuNM3^u)$.
Since $c\le N$ and $u\le NM$ we have 
$$O(cuNM3^u)=O(N NM NM 3^{NM})\le O(N^3M^23^{NM})\le 2^{O(NM+\log(NM))}.$$
Note that $NM\le (c+c')\frac{c}{c'}\binom{c+c'}{2}$.
On the interval $1\le c'\le c-1$ the
function 
$(c+c')\frac{c}{c'}\binom{c+c'}{2}$
achieves its maximum when $c'=1$, where it is 

$(c+1)c\binom{c+1}{2} \le O(c^4)$.
Hence $O(NM+\log(NM))\le O(c^4+\log c)$. 
Therefore the runtime is bounded by 
$2^{O(c^4+\log c)}$.
The space is $O(c2^{NM})$.
\item
If $N> \frac{c}{c''}\binom{c+c''}{2}$, or
$N\le \frac{c}{c''}\binom{c+c''}{2}$, then proceed
similar to the last two steps.
\end{enumerate}
\item
(This is not code. This is commentary.)
We have taken care of the cases where

$N\le c$

$N=c+1$ (this is the $c'=1$ case)

$N=c+2$ (this is the $c'=2$ case)

$\vdots$

$N=c+c-1$ (This is the $c'=c-1$ case).

Hence we have taken care of all cases where $N\le 2c-1$.
Similarly, we have taken care of all cases where $M\le 2c-1$.
Henceforth we assume $2c\le N,M$.

\item
If  $M>2\binom{2c}{2}$ or $N>2\binom{2c}{2}$
then {\it output} $\NO$ and stop.
(This works by Lemma~\ref{le:better}.)

\item
The only case left is $2c\le N,M\le 2\binom{2c}{2}=O(c^2)$.
By Lemma~\ref{le:dyn} we can determine if $\chi$ can be extended
in time $O(cuNM3^u)$. Since $u=NM=O(c^4)$ we have time 

$$O(c(NM)^2 3^u) = O(c c^8 3^{c^4}) = 2^{O(c^4 + log c)}.$$

\end{enumerate}

Step 2 and Step 4 together take time $O(N^2M^2) + 2^{O(c^4+\log c)}$.
\end{proof}

\subsection{Polynomial Kernels}

\begin{definition}
Let $A$ be a set in FPT, with parameter $c$.
$A$ has a {\it polynomial kernel} if there is a polynomial time (in the length
of the input) algorithm 
that takes input $P$ and
produces a
new problem $P'$ such that
\begin{enumerate}
\item
$P\in A$ iff $P'\in A$.
\item
The size of $P'$ is bounded by a function of $c$.
\end{enumerate}
\end{definition}

Our algorithms for $\GCE_c$ took the input and either solved the problem easily
or concluded that the problem had size bounded by a polynomial in $c$. 
Hence our algorithms showed that $\GCE_c$ has a polynomial kernel. 

\subsection{Time and Space in the Real World}\label{se:real}

We have shown that $\GCE_c$ can be computed in time $O(N^2M^2)+2^{O(c^4+\log c)}$.
If the partially-filled grid has $u$ empty spaces then the space is $O(c\times 2^u)$. 
Hence if the algorithm is run on the empty grid, so $u=NM$, the space is
$O(c2^{NM})$. 

We now look at what happens if $N=M=17$, $c=4$, and we start with the empty grid.

\noindent
{\bf Time} Even for small $c$ the additive term $2^{O(c^4+\log c)}$ is the real time-sink.
We generously assume the O-of term has constant 1 to get that the time is
$2^{4^4+\log 4} = 2^{258}\sim 10^{77}$. We generously assume that every step takes
one nano-second. Note that one nanosecond 
is $10^{-9}$ seconds.
Hence the time is $10^{77}\times 10^{-9} = 10^{68}$ seconds. This is over $2^{200}$ years. 

\smallskip

\noindent
{\bf Space} 
We generously assume the O-of term has constant 1 to get that the space is
$4\times 2^{17\times 17}=4\times 2^{289}$. This is roughly $10^{87}$ which is
larger than Eddington's estimate of the number of protons in the universe ($10^{80}$). 

A cleverer algorithm that reduces the time or space is desirable.
By Theorem~\ref{th:npc}
the time cannot be made polynomial unless $\penp$.

\section{Open Problems}\label{se:open}

We reiterate briefly the open problems stated in Section~\ref{se:no}
and add some new ones.

\begin{enumerate}
\item
The problem we really want to study is the grid coloring extension problems
with the empty grid. As noted in Section~\ref{se:no} this problem is
a sparse set, and such sets cannot be $\NP$-complete (unless $\penp$).
What is needed is a framework for proving that some sparse sets are likely not in $\P$. 

\item
In Theorem~\ref{th:fpt4} we showed that $\GCE_c$ can be computed in time $O(N^2M^2)+2^{O(c^4+\log c)}$.
Can this be improved? The last term cannot be a polynomial in $c$ (unless $\penp$);
however, it is plausible that a smaller exponential will suffice.
Is there a proof that, under assumptions,
the exponent cannot be lowered?
What is needed is a framework to prove some FPT problems are hard.

\item
Even without a theoretical improvement to our FPT algorithms, are there
heuristics one can use to speed them up in practice?

\item
We have studied grid colorings that avoid monochromatic {\it rectangles}.
One can study avoiding monochromatic {\it squares}.
The following is known:
\begin{enumerate}
\item
By a corollary to the Gallai-Witt theorem (itself a generalization of van der Waerden's theorem):
for all $c$ there exists $N=N(c)$ such that, for all $c$-coloring of $\GNN$, there is a
monochromatic square (all corners the same color). The proof gives an enormous upper bound 
even for $N(2)$; however, in reality $N(2)$ may be smaller, as we will see in the next few points.
\item
Bacher and Eliahou~\cite{BE-2010} showed the following:
\begin{enumerate}
\item
There is a 2-coloring of $G_{13,\infty}$ that has no monochromatic squares.
\item
There is a 2-coloring of $G_{14,14}$ that has no monochromatic squares.
\item
For any 2-coloring of $G_{14,15}$, there is a monochromatic square.
\item
Hence the obstruction set is $\{ G_{14,15}, G_{15,14} \}$.
\end{enumerate}
\end{enumerate}

With this in mind, we pose the following open question: is the following set $\NP$-complete:
$$
\{
(N,M,c,\chi)\st\chi
\hbox{ is extendable to a $c$-coloring of $\GNM$ with no monochromatic squares}
\}
.$$

\item
One can also look at other shapes to avoid have monochromatic. 
\end{enumerate}

\section{Acknowledgements}
We thank Amy Apon, Doug Chen, Jacob Gilbert, Matt Kovacs-Deak, Stasys Junka, Jon Katz,
Clyde Kruskal, Nathan Hayes, Erika Melder, Erik Metz, 
and Rishab Pallepati for proofreading and discussion.

We thank Wing Ning Li for pointing out that the case of
$N,M$ binary, while it seems to not be in $\NP$, is actually unknown.

We thank Jacob Gilbert, David Harris, and Daniel Marx for pointing out many improvements in the
fixed parameter algorithm which we subsequently used.

We thank
Tucker Bane,
Richard Chang,
Peter Fontana,
David Harris,
Jared Marx-Kuo,
Jessica Shi,
and
Marius Zimand,
for listening to Bill present these results and hence clarifying them.

We thank the referee for {\it many} helpful comments including a complete
reworking of the proof of Theorem~\ref{th:npc}.


%


%
%

\end{document}